\documentclass[Royal,times]{sagej_NoJournalName}
\usepackage{moreverb,url}

\usepackage[colorlinks,bookmarksopen,bookmarksnumbered,citecolor=red,urlcolor=red]{hyperref}

\newcommand\BibTeX{{\rmfamily B\kern-.05em \textsc{i\kern-.025em b}\kern-.08em
T\kern-.1667em\lower.7ex\hbox{E}\kern-.125emX}}

\newtheorem{theorem}{Theorem}[section]

\newtheorem{proposition}[theorem]{Proposition}
\newtheorem{corollary}[theorem]{Corollary}
\theoremstyle{definition}
\newtheorem{definition}[theorem]{Definition}

\theoremstyle{remark}

\begin{document}

\runninghead{Modelling Financial Market Imperfection Using Open Quantum Systems}

\title{Modelling Financial Market Imperfection Using Open Quantum Systems}

\author{Will Hicks\affilnum{1}}

\affiliation{\affilnum{1} Centre for Quantum Social and Cognitive Science, Memorial University of Newfoundland}

\corrauth{Will Hicks}

\email{whicks7940@googlemail.com}

\begin{abstract}
We start with the idea that open quantum systems can be used to represent financial markets by modelling events from the external environment and their impact on the market price. We show how to characterize distinct orbits of the time evolution, and look at the development of the reduced density matrix, that represents the state of the market, over long time frames. In particular we distinguish between classical and non-classical modes of time evolution. We show that whilst both tend to the same set of maximum entropy states, this occurs faster in classical systems, with a knock on effect on the resulting probability distributions. We demonstrate how non-classical modes of time-evolution can be used to incorporate factors such as illiquid trades and imperfect trading mechanisms, and distinguish between different mechanisms of non-classical time evolution.
\end{abstract}
\maketitle
\section{Introduction}
In \cite{HicksPHD}, the author looks at ways in which open quantum systems can be used to model the interaction between the financial market and its' environment. In particular, the author looks at how a change in the prevailing risk appetite in the external environment can impact prices. For example, positive economic news regarding a sector may lead to an increase in equity prices for that sector. Alternatively a disappointing earnings announcement may lead to a fall in the stock price. In \cite{HicksPHD}, and \cite{HicksOQS}, the author shows how such events that occur outside the market, can be represented using a simple interaction between one Hilbert space representing the market, and a second representing the environment. In addition to providing new insight into the underlying mechanism for the kind of diffusion based modelling often applied to mathematical finance, the open quantum systems approach easily generalizes to modes of diffusion that have no classical equivalent. In particular the author emphasises the crucial role that the information entropy plays in determining the nature of the time evolution, and the resulting probability distributions. The author gives examples of non-classical modes of diffusion and uses numerical simulation to study the behaviour. In this article we attempt to understand these non-classical modes for diffusion in more depth, and to consider how they can be applied to the modelling of financial market imperfections.

In section \ref{CvsNC} we show how the time evolution occurs in distinct orbits that are fixed by the initial reduced density matrix. We also look at the stationary points within each of these orbits that act as stable attractors and represent the unique equilibrium point for the relevant system. The discussion and analysis in section \ref{sec_num} shows that whilst the long term equilibrium state does not depend on the nature of the time evolution, the rate at which the system approaches equilibrium is dependent on whether the underlying mode is classical or non-classical.

In section \ref{sec_precis} we discuss how non-classical modes for time evolution can be used to represent market imperfections, and also consider how non-classical observables can play a role in addition to the non-classical time evolution of the market reduced density matrix. We also look at how to derive Kolmogorov backward equations for the time evolution of the expected value for an observable, and show that systems where at least one of the density matrix or observable is diagonalized relative to the choice of basis that simplifies the presentation of the time evolution, are fundamentally simpler to understand than those where this does not apply.
\section{Modelling Exogenous Prices Changes Using Open Quantum Systems:}\label{Background}
In this section, we give a brief overview of the analysis presented in \cite{HicksPHD} and \cite{HicksOQS}. This helps to provide context for the discussion that follows.
\subsection{Hilbert Space Setup:}\label{HilbertSpaceSetup}
In this article, for simplicity, we work in the discrete case, and represent the market Hilbert space using:
\begin{align}\label{HilbertMkt}
\mathcal{H}_{mkt}&=\mathbb{C}^N
\end{align}
The price operator acting on this Hilbert space is taken to be:
\begin{align}\label{std_price}
X&=\sum_{i=1}^N x_i|f_i\rangle\langle f_i|
\end{align}
Where $|f_i\rangle$ represent the orthonormal basis for $\mathbb{C}^N$, and $x_i$ represent the measured price values ($x_i>x_j$ for $i>j$). The environment Hilbert space is given by:
\begin{align}\label{HilbertEnv}
\mathcal{H}_{env}&=\mathbb{C}^K\otimes L^2[\mathcal{K}]
\end{align}
Here the eigenvectors in $\mathbb{C}^K$ represent different levels of environment risk appetite, and $L^2[\mathcal{K}]$ represents the Hilbert space of square integrable functions over a compact space $\mathcal{K}$. The system Hamiltonian is given by:
\begin{align}\label{HamSys}
H_{sys}&=\mathbb{I}\otimes H_{env}+H_I\\
H_I&=\sqrt{\kappa\gamma}\sum_{\alpha\in\{u,d\}}A_{\alpha}\otimes B_{\alpha}\nonumber\\
\mathbb{I}\otimes H_{env}&=\mathbb{I}\otimes\gamma\sum_{l=1}^Kl|e_l\rangle\langle e_l|\otimes H'\nonumber
\end{align}
Where $H'$ acts on $L^2[\mathcal{K}]$. $B_u/B_d$ and $A_u/A_d$ are described in section \ref{exo}.
\subsubsection{Choice of the System Hamiltonian:}
The action of the system Hamiltonian on the market Hilbert space: $\mathcal{H}_{mkt}$ defines the drift. For example, if we set $\mathcal{H}_{mkt}=L^2(\mathbb{R})$, we could write:
\begin{align*}
H_{sys}&=-ir\frac{\partial}{\partial x}\otimes H_{env}+H_I
\end{align*}
to give the wave function a constant drift at rate $r$. However, for the sake of simplicity, and since in most cases we are working with Martingale probability measures, we assume zero drift. As described in \cite{HicksPHD} section 5.4, the choice of the environment space is partly pragmatic. Ie to ensure the convergence of the necessary time integrals. However, from a financial perspective the operator:
\begin{align*}
\sum_{l=1}^Kl|e_l\rangle\langle e_l|
\end{align*}
Acts on the different levels of risk appetite, returning an escalating scalar value to the different risk appetite levels. The operator: $H'$ acts on the space: $L^2[\mathcal{K}]$ to configure the ease with which the system switches between these levels.
\subsection{Exogenous Price Moves:}\label{exo}
The interaction Hamiltonian given by $H_I$ in \ref{HamSys} is given by:
\begin{align*}
H_I&=A_u\otimes B_d+A_d\otimes B_d
\end{align*}
Here $B_u$ and $B_d$ are intended as operators that represent a jump up or down in the environment risk appetite (or pieces of good/bad news respectively). We write:
\begin{align}\label{Bud}
B_u&=\sum_{i=1}^{K-1}|e_{i+1}\rangle\langle e_i|\otimes\mathbb{I}\text{, }B_d=\sum_{i=1}^{K-1}|e_i\rangle\langle e_{i+1}|\otimes\mathbb{I}
\end{align}
The operators $A_u$ and $A_d$ represent the market reaction. For example the reaction to a piece of good news/an increase in the overall environment risk appetite (denoted $A_u$) might be a jump in the underlying price. Similarly the reaction to bad news, might be a fall in the underlying price:
\begin{align}\label{Aud}
A_u&=\sum_{i=1}^{N-1}|f_{i+1}\rangle\langle f_i|\text{, }A_d=\sum_{i=1}^{N-1}|f_i\rangle\langle f_{i+1}|
\end{align}
Since, we can make observations only on the market Hilbert space, and not the external environment, in the Open Quantum Systems approach we look at the evolution of the reduced density matrix that is obtained after taking the partial trace over the environment space:
\begin{align*}
\rho_{mkt}&=Tr_{env}[\rho]
\end{align*}
\subsection{Markovian Approximation:}
In \cite{HicksPHD} proposition 5.6.1, and \cite{HicksOQS} proposition 6.1, the author derives a Markovian approximation to the dynamics of the reduced density matrix given by:
\begin{align}\label{BMA}
\frac{d\rho_{mkt}(t)}{dt}&=-Tr_{env}[H_I(t),\rho^I(0)]\\
&+\sigma^2\Big(A_u\rho_{mkt}(t)A_d+A_d\rho_{mkt}(t)A_u-\frac{1}{2}\{A_uA_d+A_dA_u,\rho_{mkt}(t)\}\Big)\nonumber\\
&+\nu_u^2\Big(A_u\rho_{mkt}(t)A_u-\frac{1}{2}\{A_uA_u,\rho_{mkt}(t)\}\Big)\nonumber\\
&+\nu_d^2\Big(A_d\rho_{mkt}(t)A_d-\frac{1}{2}\{A_dA_d,\rho_{mkt}(t)\}\Big)\nonumber
\end{align}
Here $\sigma^2$, $\nu_u^2$ and $\nu_d^2$ are determined by the environment state. We write:
\begin{align}\label{constants}
\rho(t)&=\sum_{l,m=1}^K\rho^{lm}_{mkt}(t)\otimes r_{lm}|e_l\rangle\langle e_m|\otimes\rho_B
\end{align}
Where:
\begin{align*}
\sigma^2&=\kappa Tr\Big[B_uB_d\sum_{l,m=1}^Kr_{lm}|e_l\rangle\langle e_m|\otimes\rho_B\Big]=\kappa\sum_{l=1}^{K-1}r_{ll}\\
\nu_u^2&=\kappa Tr\Big[B_uB_u\sum_{l,m=1}^Kr_{lm}|e_l\rangle\langle e_m|\otimes\rho_B\Big]=2\kappa\sum_{l=1}^{K-2}r_{l(l+2)}\\
\nu_d^2&=\kappa Tr\Big[B_dB_d\sum_{l,m=1}^Kr_{lm}|e_l\rangle\langle e_m|\otimes\rho_B\Big]=2\kappa\sum_{l=1}^{K-2}r_{(l+2)l}
\end{align*}
Taking $N\rightarrow\infty$, so that we can ignore boundary conditions, we can re-write equation \ref{BMA}:
\begin{proposition}\label{nonclass_prop}
Assume $N\rightarrow\infty$. Then we have:
\begin{align}\label{BMA_nonclass}
\frac{d\rho_{mkt}(t)}{dt}&=-Tr_{env}[H_I(t),\rho^I(0)]\\
&+\sigma^2\mathcal{L}(\rho_{mkt}(t))-\frac{\nu_u^2}{2}\mathcal{L}(A_u\rho_{mkt}(t)A_u)-\frac{\nu_d^2}{2}\mathcal{L}(A_d\rho_{mkt}(t)A_d)\nonumber\\
\text{Where:}\nonumber\\
\mathcal{L}(M)_{ij}&=M_{(i-1)(j-1)}+M_{(i+1)(j+1)}-2M_{ij}\nonumber
\end{align}
\end{proposition}
\begin{proof}
The result follows from noting that, in the limit $N\rightarrow\infty$, we have:
\begin{align*}
A_uA_u|f_i\rangle\langle f_j|&=A_u\big(A_u|f_i\rangle\langle f_j|A_u\big)A_d\\
|f_i\rangle\langle f_j|A_uA_u&=A_d\big(A_u|f_i\rangle\langle f_j|A_u\big)A_u\\
A_dA_d|f_i\rangle\langle f_j|&=A_d\big(A_d|f_i\rangle\langle f_j|A_d\big)A_u\\
|f_i\rangle\langle f_j|A_dA_d&=A_u\big(A_d|f_i\rangle\langle f_j|A_d\big)A_d
\end{align*}
\end{proof}
\begin{proposition}
The mapping represented by equation \ref{BMA} is completely positive if and only if $\nu_u^2+\nu_d^2\leq\sigma^2$
\end{proposition}
\begin{proof}
It follows from \cite{Chrus} Theorem 1 (originally \cite{GKS} Theorem 2.2) that a linear operator $\mathcal{L}:\mathbb{M}_N\rightarrow\mathbb{M}_N$ is the generator of a completely positive semigroup of $\mathbb{M}_N$ if it can be expressed in the following form:
\begin{align}\label{GKS_th_eq}
\mathcal{L}\rho&=-[H,\rho]+\frac{1}{2}\sum_{k,l=1}^{N^2-1}c_{kl}\Big([F_k,\rho F_l^*]+[F_k\rho,F_l^*]\Big)
\end{align}
Where $H=H^*$, $TrH=0$, $TrF_k=0$, $Tr[F_kF_l^*]=\delta_{kl}$ and $c_{kl}$ is a complex positive matrix. One can write equation \ref{BMA} in this form by setting:
\begin{align}\label{GKS_for_eq7}
F_1&=A_u\text{, }F_2=A_d\\
c_{11}&=c_{22}=\sigma^2\text{, }c_{12}=-\nu_u^2\text{, }c_{21}=-\nu_d^2\nonumber
\end{align}
Thus we can see that equation \ref{BMA} is naturally in the GKS standard form. Furthermore, $c_{kl}$ is a positive matrix if and only if $\nu_u^2+\nu_d^2\leq\sigma^2$.
\end{proof}
\section{Classical vs Non-Classical Diffusion and Entropy Limits:}\label{CvsNC}
In this section, we start by establishing some basic results related to equation \ref{BMA}. We then go on in section \ref{class_syst} to classify what we mean by a `classical' system. In section \ref{sec_diff_orb} we study the behaviour of the time-evolution described by \ref{BMA}, first by characterizing different diffusion orbits, before looking at the long term behaviour and the nature of stationary points, before returning to the question of distinguishing classical and nonclassical diffusion in section \ref{ClassicalVsNonclassical}.
\begin{proposition}\label{gen entropy prop}
Let the market Hilbert space be finite dimensional. Then the Von-Neumann entropy for the market state under the dynamics described by \ref{BMA} is non-decreasing.
\end{proposition}
\begin{proof}
First note that it follows from the definition of the relative entropy (see for example \cite{BP} section 2.3.2) that we have:
\begin{align}\label{rel_ent_ineq}
S(\rho_{mkt}(t)||\rho_{0,mkt})&\leq S(\rho||\rho_0)
\end{align}
Where $\rho_0$ is a stationary reduced density matrix acting on $\mathcal{H}_{mkt}$, for the dynamics described by \ref{BMA}. We assume $V(t)$ is given by:
\begin{align*}
V(t)\rho_{mkt}(0)&=Tr_{env}[U(t,0)\rho_{mkt}(0)\otimes\rho_{env}U(t,0)^{\dagger}]
\end{align*}
Where $U(t,0)$ is an arbitrary unitary map. From \ref{rel_ent_ineq}, it follows (see for example \cite{BP} section 3.2.5):
\begin{align*}
S\big(V(t)\rho(0)||V(t)\rho_0\big)&=S\big(Tr_{env}[U(t,0)\rho(0)U(t,0)^{\dagger}]||Tr_{env}[U(t,0)\rho_0U(t,0)^{\dagger}]\big)\\
&\leq S\big(U(t,0)\rho(0)U(t,0)^{\dagger}||U(t,0)\rho_0U(t,0)^{\dagger}\big)\\
&=S\big(\rho(0)||\rho_0\big)
\end{align*}
Since, by assumption, $\rho_0$ is a stationary states, it follows that:
\begin{align*}
S\big(V(t)\rho(0)||V(t)\rho_0\big)&=S\big(V(t)\rho(0)||\rho_0\big)
\end{align*}
So:
\begin{align}\label{ent_proof}
S(\rho_{mkt}(t)||\rho_0)&\leq S(\rho_{mkt}(0)||\rho_0)
\end{align}
Therefore:
\begin{align}\label{ent_key_eq}
-S(\rho_{mkt}(t))-Tr[\rho_{mkt}(t)\log(\rho_0)]&\leq -S(\rho_{mkt}(0))-Tr[\rho_{mkt}(0)\log(\rho_0)]\nonumber\\
S(\rho_{mkt}(0))-S(\rho_{mkt}(t))&\leq Tr[\rho_{mkt}(t)\log(\rho_0)]-Tr[\rho_{mkt}(0)\log(\rho_0)]\nonumber\\
&=Tr[(\rho_{mkt}(t)-\rho_{mkt}(0))\log(\rho_0)]
\end{align}
We can select $\rho_0=\frac{1}{N}\mathbb{I}$ as the stationary state, whereby \ref{ent_key_eq} becomes:
\begin{align*}
S(\rho_{mkt}(0))-S(\rho_{mkt}(t))&\leq 0
\end{align*}
\end{proof}
\subsection{Classical Systems:}\label{class_syst}
A classical model in finite dimensions generally consists of a set of outcomes: $\{ x_i:i\in 1,\dots,N\}$, and an associated set of (time dependent) probabilities: $\{ p_i(t):i\in 1,\dots,N\}$. These can be represented as a diagonal density matrix acting on a diagonal observable:
\begin{definition}\label{define_class_diff}
A classical system is defined as an observable and density matrix that are simultaneously diagonalizable:
\begin{align*}
X&=\begin{bmatrix}x_{1} & & \\ & \ddots & \\ & & x_{N}\end{bmatrix}\text{, }\rho_{mkt}(t)=\begin{bmatrix}p_{1}(t) & & \\ & \ddots & \\ & & p_{N}(t)\end{bmatrix}
\end{align*}
So that:
\begin{align*}
E[X]&=\sum_{i=1}^Nx_ip_i(t)
\end{align*}
Working in the Schr{\"o}dinger interpretation, a classical diffusion is defined as a time evolution of the system represented by $\rho_{mkt}(t)$ and $X$, such that the system is a classical system at all times.
\end{definition}
In \cite{HicksPHD} proposition 5.4.13, it is shown that if we assume:
\begin{itemize}
\item the environment remains in a time independent maximum entropy state. That is:
\begin{align}\label{diag_env}
\rho(t)&=\rho_{mkt}(t)\otimes\frac{1}{K}\sum_{i=1}^K|e_i\rangle\langle e_i|\otimes\rho_B
\end{align}
\item the market density matrix begins in a diagonal state:
\begin{align}\label{diag_rho_0}
\rho_{mkt}(0)&=\sum_{i=1}^Np_i(0)|f_i\rangle\langle f_i|
\end{align}
\item $-Tr_{env}[H_I(t),\rho^I(0)]=0$.
\end{itemize}
We find that in equation \ref{BMA} we have: $\nu_u^2=\nu_d^2=0$, and the time evolution given by equation \ref{BMA} simplifies to:
\begin{align}\label{BMA_class}
\frac{d\rho_{mkt}(t)}{dt}&=+\sigma^2\Big(A_u\rho_{mkt}(t)A_d+A_d\rho_{mkt}(t)A_u-\frac{1}{2}\{A_uA_d+A_dA_u,\rho_{mkt}(t)\}\Big)
\end{align}
In this instance, we can write (\cite{HicksPHD} proposition 5.4.13):
\begin{align}\label{diagBMA}
\frac{d\rho_{mkt}(t)}{dt}&=\sigma^2\sum_{i=1}^N\Big(p_{i+1}(t)+p_{i-1}(t)-2p_i(t)\Big)|f_i\rangle\langle f_i|\\
&+\sigma^2\Big((p_2(t)-p_1(t)|f_1\rangle\langle f_1|+(p_{N-1}(t)-p_N(t))|f_N\rangle\langle f_N|\Big)\nonumber
\end{align}
We incorporate two important properties of this classical system in the following proposition:
\begin{proposition}\label{class ent}
The Von-Neummann entropy for the market state \ref{diag_rho_0}, under the dynamics described by \ref{diagBMA}, increases monotonically, tending to the limit:
\begin{align*}
\lim_{t\rightarrow\infty} Tr[\rho_{mkt}(t)\log(\rho_{mkt}(t))]&= \log(N)
\end{align*}
\end{proposition}
\begin{proof}
First note that for a classical system where the initial market state is represented by a diagonal reduced density matrix (ie \ref{diag_rho_0}) and evolving according to equation \ref{diagBMA}, the reduced density matrix remains diagonal:
\begin{align*}
\rho_{mkt}(t)&=\sum_{i=1}^Np_i(t)|f_i\rangle\langle f_i|
\end{align*}
Where $p_i(t)$ can be found by solving equation \ref{diagBMA}. See also \cite{HicksPHD} 5.4.13. Next note that, where $\mathcal{L}(t)$ is the dynamical map described by \ref{diagBMA}:
\begin{align*}
L(t)\rho_{mkt}(t)&=0\implies p_i(t)=p_j(t)\text{, for: }i\neq j
\end{align*}
Thus we have that:
\begin{align*}
\rho_0&=\frac{1}{N}\sum_{i=1}^N|f_i\rangle\langle f_i|
\end{align*}
is the only stationary `classical' state. Thus the first statement then follows from proposition \ref{gen entropy prop}.
\end{proof}
\subsection{Characterizing Diffusion Orbits:}\label{sec_diff_orb}
We see from equations \ref{BMA} and \ref{constants} that because: $r_{lm}=0$ for $l\neq m$, we have $\nu_u^2=\nu_d^2=0$ and the terms:
\begin{align*}
\nu_u^2\Big(A_u\rho_{mkt}(t)A_u-\frac{1}{2}\{A_uA_u,\rho_{mkt}(t)\}\Big)+\nu_d^2\Big(A_d\rho_{mkt}(t)A_d-\frac{1}{2}\{A_dA_d,\rho_{mkt}(t)\}\Big)
\end{align*}
do not contribute to the time evolution. As a result, we can think of equation \ref{BMA_class} as shifting probability along matrix diagonals, and a classical density matrix will remain classical at all times. For this reason, we term this {\em classical} diffusion.

In the next two propositions, we characterize diffusion orbits (proposition \ref{diff_orb}) and stationary points (proposition \ref{stat_pts}) that arise from the time evolution given by equation \ref{BMA}, in the case that $\nu_u, \nu_d\neq 0$.
\begin{proposition}\label{diff_orb}
Let $\rho(t)$ be a density matrix acting on a separable Hilbert space $\mathcal{H}$, and let the time evolution of this density matrix be given by equation \ref{BMA}. First, we define the orbits $\mathcal{D}^{\overline{\epsilon}}(\mathcal{H})$:
\begin{align}\label{orb}
\rho_{mkt}(t)\in\mathcal{D}^{\overline{\epsilon}}(\mathcal{H})&\Leftrightarrow\sum_{i=1}^{N-j}\rho_{i(i+j)}(t)=\epsilon_j\text{, }\sum_{i=1}^{N-j}\rho_{(i+j)i}(t)=\epsilon_j^{\ast}
\end{align}
Then we have: 
\begin{align}\label{orb3}
\rho(t)\in\mathcal{D}^{\overline{\epsilon}}(\mathcal{H})&\implies \rho(u)\in\mathcal{D}^{\overline{\epsilon}}(\mathcal{H})\text{, for all }u\geq t
\end{align}
That is:
\begin{align*}
\frac{d}{dt}\bigg(\sum_i\rho_{i(i+j)}(t)\bigg)&=0
\end{align*}
Where, $N$ is the Hilbert space (possibly infinite) dimension.
\end{proposition}
\begin{proof}
We label the off-diagonal sums as follows:
\begin{align*}
\Sigma_j(t)&=\sum_i\rho_{i(i+j)}(t)
\end{align*}
Now consider a single matrix element:
\begin{align*}
\rho&=\rho_{kl}|f_k\rangle\langle f_l|
\end{align*}
The impact to $\Sigma_j(t)$ after a discrete time-step $\delta t$ from the term: $\mathcal{L}(A_d\rho A_d)$ is given by:
\begin{align*}
\mathcal{L}(A_d\rho A_d)&_{k(l+2)}=\rho_{kl}\\
\mathcal{L}(A_d\rho A_d)&_{(k-2)l}=\rho_{kl}\\
\mathcal{L}(A_d\rho A_d)&_{(k-1)(l+1)}=-2\rho_{kl}
\end{align*}
Similarly, the impact from the term $\mathcal{L}(A_d\rho A_d)$ is given by:
\begin{align*}
\mathcal{L}(A_u\rho A_u)&_{k(l-2)}=\rho_{kl}\\
\mathcal{L}(A_u\rho A_u)&_{(k+2)l}=\rho_{kl}\\
\mathcal{L}(A_u\rho A_u)&_{(k+1)(l-1)}=-2\rho_{kl}
\end{align*}
Finally, the impact from the term $\mathcal{L}(\rho)$ is given by:
\begin{align*}
\mathcal{L}(\rho)&_{(k-1)(l-1)}=\rho_{kl}\\
\mathcal{L}(\rho)&_{(k+1)(l+1)}=\rho_{kl}\\
\mathcal{L}(\rho)&_{kl}=-2\rho_{kl}
\end{align*}
Now summing over $j=k-l$, $j=k+2-l$ and $j=k-2-l$, we get:
\begin{align*}
\Sigma_{k-l}(t+\delta t)-\Sigma_{k-1}(t)&=\rho_{kl}+\rho_{kl}-2\rho_{kl}=0\\
\Sigma_{k+2-l}(t+\delta t)-\Sigma_{k+2-l}(t)&=\rho_{kl}+\rho_{kl}-2\rho_{kl}=0\\
\Sigma_{k-2-l}(t+\delta t)-\Sigma_{k-2-l}(t)&=\rho_{kl}+\rho_{kl}-2\rho_{kl}=0
\end{align*}
The result then follows from summing over matrix elements $\rho_{kl}$.
\end{proof}
We now go on to classify the stationary points for the system, before showing how these can be used to understand the long term behaviour for time evolution. Specifically we show that the stationary points are stable attractors, in sense of definition \ref{stable_att}, and that each diffusion orbit has a unique limit point. First, we define a stable attractor as follows:
\begin{definition}\label{stable_att}
A stable attractor under a mapping $\phi$ is defined as a density matrix $\rho_0$ such that:
\begin{align*}
||\phi(\rho)-\rho_0||<||\rho-\rho_0||
\end{align*}
Where $||\text{...}||$ represents the Frobenius norm.
\end{definition}
\begin{proposition}\label{stat_pts}
Assume that the dynamics for market reduced density matrix acting on $\mathbb{C}^N$ (where $N<\infty$), are described by equation \ref{BMA}. Then the stationary points are given by the Toeplitz matrices $T^{\overline{\epsilon}}$:
\begin{align}\label{Toeplitz}
T^{\overline{\epsilon}}&=\frac{1}{N}\mathbb{I}+\sum_{j=1}^{N-1}\sum_{i=1}^{N-j}\bigg(\frac{\epsilon_j}{N-j}\bigg)|f_{i+j}\rangle\langle f_i|+\bigg(\frac{\epsilon_j}{N-j}\bigg)^{\ast}|f_i\rangle\langle f_{i+j}|
\end{align}
These stationary points are stable attractors in the sense of definition \ref{stable_att}, and are the unique limit point within each diffusion orbit: $\mathcal{D}^{\overline{\epsilon}}(\mathcal{\mathbb{C}^N})$.
\end{proposition}
\begin{proof}
That the matrices: $T^{\overline{\epsilon}}$ are stationary points follows from the fact that:
\begin{align*}
\mathcal{L}(T^{\overline{\epsilon}})&=\mathcal{L}(A_uT^{\overline{\epsilon}}A_u)\\
&=\mathcal{L}(A_dT^{\overline{\epsilon}}A_d)\\
&=0
\end{align*}
Now we consider a perturbation: $M_{\delta}$, to the Toeplitz matrix, and look at the Frobenius norm, after applying the map:
\begin{align*}
\phi(\rho;\delta t)&=\rho+\sigma^2\mathcal{L}(\rho)\delta t-\nu_u^2\mathcal{L}(A_u\rho A_u)\delta t-\nu_d^2\mathcal{L}(A_d\rho A_d)\delta t
\end{align*}
The fact that $T^{\overline{\epsilon}}$ is a stable attractor follows if we can show that:
\begin{align*}
||\phi(T^{\overline{\epsilon}}+M_{\delta})-T^{\overline{\epsilon}}||<||M_{\delta}||
\end{align*}
Now, by linearity we have:
\begin{align*}
||\phi(T^{\overline{\epsilon}}+M_{\delta})-T^{\overline{\epsilon}}||&=||\phi(T^{\overline{\epsilon}})+\phi(M_{\delta})-T^{\overline{\epsilon}}||\\
&=||\phi(M_{\delta})||
\end{align*}
So we must show that:
\begin{align*}
||\phi(M_{\delta})||&<||M_{\delta}||
\end{align*}
We write:
\begin{align*}
M_{\delta}&=\sum_{i,j}M_{ij}|f_i\rangle\langle f_j|
\end{align*}
Then:
\begin{align*}
\phi(M_{\delta})&=\phi\Big(\sum_{i,j}M_{ij}|f_i\rangle\langle f_j|\Big)\\
&=\sum_{i,j}\phi(M_{ij}|f_i\rangle\langle f_j|)\\
\end{align*}
So, by the triangle inequality:
\begin{align*}
||\phi(M_{\delta})||&\leq\sum_{i,j}||\phi(M_{ij}|f_i\rangle\langle f_i|)||
\end{align*}
Therefore, the result follows from showing: $||\phi(M_{ij}|f_i\rangle\langle f_i|)||\leq M_{ij}$
Furthermore, we have:
\begin{align*}
||\phi(M_{ij}|f_i\rangle\langle f_i|)||^2&=1-4\sigma^2\delta t+O(\delta t^2)\\
\leq M_{ij}^2
\end{align*}
Therefore $||\phi(T^{\overline{\epsilon}}+M_{\delta})-T^{\overline{\epsilon}}||<||M_{\delta}||$ as required. Since $T^{\overline{\epsilon}}$ is the only stationary point in the orbit: $\mathcal{D}^{\overline{\epsilon}}(\mathbb{C}^N)$, it follows that $T^{\overline{\epsilon}}$ is the unique limit point within this diffusion orbit.
\end{proof}
\subsection{Classical vs Non-Classical Diffusion:}\label{ClassicalVsNonclassical}
It is well known that the general form for a completely positive and trace preserving quantum Markov semigroup for a density matrix acting on a separable Hilbert space can be written in the form (see for example \cite{Chrus} Theorem 2, originally \cite{Lindblad} Theorem 2):
\begin{align}\label{GKSL}
\frac{d\rho(t)}{dt}&=-i[\hat{H},\rho(t)]+\sum_i\gamma_i\Big(L_i\rho(t)L_i^{\dagger}-\frac{1}{2}\{L_iL_i^{\dagger},\rho(t)\}\Big)
\end{align}
We now apply this standard form to determine when a time evolution can be thought of as classical vs non-classical. First we give the following definition:
\begin{definition}\label{define_nonclass_diff}
A time evolution is considered classical in the case whereby when applied to a classical system in the sense of definition \ref{define_class_diff}, the system remains classical.
\end{definition}
\begin{proposition}\label{Lindblad}
Let $\rho(t)$ be a density matrix acting on a finite dimensional Hilbert space $\mathcal{H}=\mathbb{C}^N$. We write:
\begin{align*}
\rho(t)&=\sum_i\rho_{ij}(t)|f_i\rangle\langle f_j|
\end{align*}
A quantum Markov semigroup acting on a density matrix given by equation \ref{GKSL} is a classical diffusion, per definition \ref{define_nonclass_diff}, if and only if the $L_i$ are orthogonal. Furthermore, the sets of matrix elements defined as follows:
\begin{align}\label{diff_orb2}
\mathcal{D}_j&=\sum_{i=1}^{N-j}\rho_{i(i+j)}(t)|f_i\rangle\langle f_{i+j}|\text{, }\mathcal{D}_j^*=\sum_{i=1}^{N-j}\rho_{(i+j)i}(t)|f_{i+j}\rangle\langle f_i|
\end{align}
are closed under the map defined by equation \ref{GKSL} if and only if $L_i$ are orthogonal matrices.
\end{proposition}
\begin{proof}
If $L_i$ is orthogonal, we have (for $\alpha_k\in\mathbb{R}$):
\begin{align*}
L_i|f_j\rangle&=\alpha_k|f_k\rangle\\
L_iL_i^{\dagger}&=\sum_k\alpha_k|f_k\rangle\langle f_k|
\end{align*}
for some orthonormal basis. Then, by relabelling, and scaling if necessary, we can write:
\begin{align*}
L_i|f_k\rangle&=|f_{k+n}\rangle
\end{align*}
We now consider the impact of equation \ref{GKSL} on an arbitrary matrix element $|f_k\rangle\langle f_{k+l}|\in\mathcal{D}_l$:
\begin{align*}
L_i|f_k\rangle\langle f_{k+l}|L_i^{\dagger}&=|f_{k+n}\rangle\langle f_{k+l+n}|\\
L_iL_i^{\dagger}|f_k\rangle\langle f_{k+l}|&=|f_k\rangle\langle f_{k+l}|L_iL_i^{\dagger}\\
&=|f_k\rangle\langle f_{k+l}|
\end{align*}
Thus it follows that each set of elements $\mathcal{D}_l$ is closed under the mapping defined by \ref{GKSL}. Now assume there is an operator $L_i$, that does not fulfil the orthogonality condition. Then we can write:
\begin{align*}
L_i|f_j\rangle&=\alpha_k|f_k\rangle+\alpha_l|f_l\rangle
\end{align*}
Where $\alpha_k,\alpha_l\neq 0$. Now we have:
\begin{align*}
L_i|f_k\rangle\langle f_{k+l}|L_i^{\dagger}&=\alpha_k\alpha_l(|f_k\rangle\langle f_l|+|f_l\rangle\langle f_k|)+\alpha_k^2|f_k\rangle\langle f_k|+\alpha_l^2|f_l\rangle\langle f_l|
\end{align*}
Whereupon it follows that the sets: $\mathcal{D}_l$ are not closed.
\end{proof}
\begin{corollary}
The time evolution described by equation \ref{BMA} is a classical time evolution, according to definition \ref{define_nonclass_diff}, if and only if $\nu_u^2=\nu_d^2=0$.
\end{corollary}
\begin{proof}
We can write equation \ref{BMA} in the Lindblad form using:
\begin{align*}
L_1&=\cos(\theta)A_u-\sin(\theta)A_d\\
L_2&=\sin(\theta)A_u+\cos(\theta)A_d
\end{align*}
Where:
\begin{align*}
\begin{pmatrix}\cos(\theta)&-\sin(\theta)\\ \sin(\theta)&\cos(\theta)\end{pmatrix}\begin{pmatrix}\sigma^2&-\nu_u^2\\-\nu_d^2&\sigma^2\end{pmatrix}\begin{pmatrix}\cos(\theta)&+\sin(\theta)\\ -\sin(\theta)&\cos(\theta)\end{pmatrix}&=\begin{pmatrix}\sigma_1^2&0\\0&\sigma_2^2\end{pmatrix}
\end{align*}
For some $\sigma_1,\sigma_2$. Therefore, the result follows from the fact that $L_1$ and $L_2$ are orthogonal if and only if $\nu_u=\nu_d=0$. Indeed:
\begin{align*}
L_1L_1^T&=\cos^2(\theta)A_uA_d+\sin^2(\theta)A_dA_u-\sin(\theta)\cos(\theta)A_dA_d-\sin(\theta)\cos(\theta)A_uA_u
\end{align*}
\end{proof}
\section{Market Imprecision:}\label{sec_precis}
\subsection{Modelling Market Imprecision Using Non-Classical Observables:}\label{precision}
There are some instances where an asset price is known with a high degree of certainty in advance. Where this is not the case, we consider two possible reasons why not:
\begin{enumerate}
\item[a)] The purchase is to occur at some time in the future, and one does not know what will happen to impact the price between now and then.
\item[b)] There are imperfections in the market mechanism for executing a trade.
\end{enumerate}
\underline{Case a): Price Evolution via External Influence}\newline
\noindent The first case is well handled by the time evolution described by equation \ref{BMA}. Essentially, the interaction between the market, and the external environment is boiled down to:
\begin{itemize}
\item Good news event, risk appetite \& price increases: $A_u\otimes B_u$.
\item Bad new event, risk appetite \& price decreases: $A_d\otimes B_d$.
\end{itemize}
As illustrated in \cite{HicksPHD} and \cite{HicksOQS}, the possibility of such events leads to a probability distribution for the future price. This uncertainty is down to an increase in the market entropy: essentially case a) above.\newline

\noindent\underline{Case b): Imperfect Market Mechanism}\newline
Where one executes a small trade in a market with reasonable liquidity, one may be able to determine the execution price (to a close approximation) prior to execution. For example by looking up the current market quotes using one of the standard data aggregation platforms available (Bloomberg, Reuters etc). We can model such an eventuality using the market state:
\begin{align}\label{dirac}
\rho_{mkt}(0)&=|f_k\rangle\langle f_k|
\end{align}
Throughout this article, we assume the standard price operator is given by equation \ref{std_price}. Under the market state given by equation \ref{dirac}, this will return the price $x_k$ with certainty. We will now consider an operator representing a trade which for whatever reason, is slightly misaligned with the market. For example, one might consider a contract with non-standard terms, or a large trade in an illiquid stock. We can model this situation by changing the eigenvectors for the relevant price operator:
\begin{align}\label{NC_op}
X&=\sum_{i=1}^Nx_i|v_i\rangle\langle v_i|
\end{align}
There is now uncertainty in the measured outcome for the non-standard trade that results from an imperfect market mechanism, rather than from a lack of information. Even when the market is in a Dirac state, \ref{dirac}, we may still find non-zero variance. For example we set:
\begin{align}\label{v_vec}
|v_i\rangle&=\epsilon|f_{i-1}\rangle+\sqrt{1-2\epsilon^2}|f_j\rangle+\epsilon|f_{i+1}\rangle
\end{align}
Then we have, where $\rho_{mkt}$ is given by equation \ref{dirac}, the variance is given by:
\begin{align*}
Var(X)&=E[X^2]-x_k^2\\
&=Tr[\sum_ix_i^2|v_i\rangle\langle v_i|f_k\rangle\langle f_k]-x_k^2\\
&=\epsilon^2x_{k+1}^2+\epsilon^2x_{k-1}^2-2\epsilon^2x_k^2
\end{align*}
Where if: $x_{k+1}=x_k+\delta x$ and $x_{k-1}=x_k-\delta x$, we find that:
\begin{align*}
Var(X)&=2\epsilon^2\delta x^2
\end{align*}
\subsection{Kolmogorov Backward Equations:}
For many practical purposes one is interested in the evolution of the expected value for some function of the main price observable, as we look further into the future. Therefore, one way of investigating the impact of non-classical observables, such as that given in equation \ref{NC_op}, is by considering the time evolution of $u(t)=E^t[f(X)]$.
\begin{proposition}\label{KBE_prop}
We define $u(t,x)$ as follows:
\begin{align*}
u(t)&=E^t[f(X)]
\end{align*}
Then, under the time evolution given by equation \ref{BMA}, the evolution of $u(t)$ is given by:
\begin{align}\label{KBE}
\frac{du}{dt}&=Tr\Big[\rho_{mkt}(t)\sum_{k,l=u,d}c_{kl}A_kf(X)A_l-\frac{1}{2}\{A_kA_l,f(X)\}\Big]
\end{align}
Where: $c_{ud}=c_{du}=\sigma^2$, and $c_{uu}=\nu_u^2$, $c_{dd}=\nu_d^2$.
\end{proposition}
\begin{proof}
Since we are working in the Schr{\"o}dinger interpretation, it follows that:
\begin{align*}
\frac{du}{dt}&=Tr\Big[\frac{d\rho_{mkt}(t)}{dt}f(X)\Big]\\
&=Tr\Big[\Big(\sum_{k,l=u,d}c_{kl}A_k\rho_{mkt}(t)A_l-\frac{1}{2}\{A_kA_l,\rho_{mkt}(t)\}\Big)f(X)\Big]
\end{align*}
By cycling the operators we end up with equation \ref{KBE}, whereby we are working in the Heisenberg interpretation. That it, the time evolution is applied to the observable $f(X)$, rather than to the market reduced density matrix: $\rho_{mkt}(t)$.
\end{proof}
\subsection{Type I Non-Classical Systems:}
Comparing equations \ref{BMA} and \ref{KBE}, it becomes clear that assuming the system starts as a classical system (per definition \ref{define_class_diff}), then even if $\nu_u,\nu_d\neq 0$ then one of these 2 will remain diagonal. That is:
\begin{itemize}
\item We can work in the Schr{\"o}dinger interpretation, and apply equation \ref{BMA}. The market reduced density matrix becomes non-diagonal if $\nu_u,\nu_d\neq 0$. $f(X)$ remains diagonal.
\item We can work in the Heisenberg interpretation, and apply equation \ref{KBE}. That is the market reduced density matrix remains diagonal, and the operator $f(X)$ evolves to become non-diagonal if $\nu_u,\nu_d\neq 0$.
\end{itemize}
Note that in either case, whilst the matrix elements: $\rho_{ij}(t)$ for $i\neq j$, help to determine the time evolution, they play no role in taking the expectation. If we have:
\begin{align*}
\rho_{mkt}(t)&=\sum_{i,j}\rho_{ij}(t)|f_i\rangle\langle f_j|\\
f(X)&=\sum_if(x_i)|f_i\rangle\langle f_i|
\end{align*}
Then:
\begin{align}\label{NC_typeI}
E^t[f(X)]&=\sum_i\rho_{ii}(t)f(x_i)
\end{align}
That is the diagonal elements of $\rho_{mkt}(t)$ can be interpreted as `classical' probabilities. Ie $\rho_{ii}(t)$ is the probability of measuring $x_i$ at time $t$. This leads to the following definition:
\begin{definition}\label{TypeIDef}
A type I non-classical system is defined as an observable: $X$, a density matrix: $\rho$, and a time evolution such as that defined by equation \ref{BMA}, such that at least one of the observable or the density matrix is diagonalized relative to the basis defined by the time evolution.
\end{definition}
Note that for a type I non-classical system, per definition \ref{TypeIDef}, we have (using the notation given by equation \ref{diff_orb2}):
\begin{itemize}
\item $\mathcal{D}_0$ are the set of probabilities for the eigenvalues of a diagonal operator of the type given by equation \ref{std_price}.
\item The elements in $\mathcal{D}_i$, for $i\neq 0$, do not directly impact the probabilities. However, they will impact the time-evolution of elements in $\mathcal{D}_0$.
\end{itemize}
\subsection{Type II Non-Classical Systems:}
\begin{definition}\label{TypeIIDef}
We classify type II non-classical systems as being any system consisting of an observable: $X$, a reduced density matrix: $\rho_{mkt}(t)$, and a time evolution of the type given by equation \ref{BMA}, whereby the assumptions of definitions \ref{define_class_diff} and \ref{TypeIDef} do not apply. That is both the observable, and the reduced density matrix are non-diagonal in the basis defined by the time evolution.
\end{definition}
As noted above, the time evolution, and in particular the simple form in terms of being a combination of orthogonal operators: $A_u$ and $A_d$, fixes the choice of basis. If both $\rho_{mkt}(t)$ and $f(X)$ are non-diagonal in this basis then the diagonal elements in $\mathcal{D}_0$ can no longer be interpreted as probabilities. That is:
\begin{align*}
E^t[f(X)]&=\sum_i\rho_{ii}f(X)_{ii}
\end{align*}
becomes:
\begin{align*}
E^t[f(X)]&=\sum_{i,j}\rho_{ij}f(X)_{ij}
\end{align*}
Whilst we can diagonalize $\rho_{mkt}(t)$, applying the necessary unitary transformation destroys the simple form for the time evolution. We illustrate this using a simple one step simulation.
\subsubsection{Simulating Type II Non-Classical Systems:}
\begin{itemize}
\item We assume the Hilbert space is given by $\mathbb{C}^N$, with $N=21$.
\item We start with a Dirac state: $\rho_0=|f_{11}\rangle\langle f_{11}|$.
\item We set $x_1=-1$, $x_{21}=1$, and assume evenly spaced values: $x_{i+1}=x_i+0.1$.
\item We carry out a single time-step, with $\delta t=0.01$, $\sigma=0.4$, and $\nu_u=\nu_d=0.36$.
\item We assume that $\rho_0$ has been diagonalized by a random unitary matrix: $U$. That is:
\begin{align*}
A_u&\rightarrow U^*A_uU\text{, }A_d\rightarrow U^*A_dU
\end{align*}
Whilst this is unrealistic, it serves to illustrate the impact that the loss of simplicity in the time evolution has.
\end{itemize}
Figure \ref{nice} shows the non-zero matrix elements for a standard simulation step. Ie, a type I simulation where there has been no pre-diagonalization using a random unitary matrix. The plot shows the non-central matrix elements. We exclude the large central element $\rho_{10,10}=0.9968$, so as to highlight the regular structure in this case.

Figure \ref{not} shows the same results for the simulation where we assume the starting reduced density matrix has been diagonalized by a random unitary matrix. There are 2 things to note. First that the regular structure of the off-diagonal elements has been lost in figure \ref{not}. Secondly, the magnitude of the off-diagonal elements are much smaller. This is because, figure \ref{nice} shows all non-zero elements, which spill out from the central point. In the type II simulation non-zero points are spread throughout the whole reduced density matrix even after the first time-step. It is worth emphasizing again that the use of the random unitary matrix is unlikely to be of practical significance. It simply serves to illustrate the difference between the type I, and type II non-classical systems.
\begin{figure}
\includegraphics[scale=1]{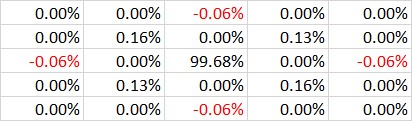}
\includegraphics[scale=0.86]{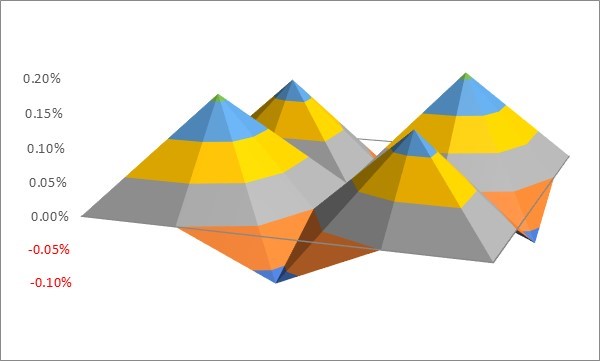}
\caption{Reduced density matrix after a single time-step, where we assume the reduced density matrix starts in a Dirac state. The plot shows the non-central matrix elements only. That is it does not depict the large central matrix element of $0.9968$. This helps to illustrate the regularity that is due to the simple form for the time-evolution.}\label{nice}
\end{figure}
\begin{figure}
\includegraphics[scale=1]{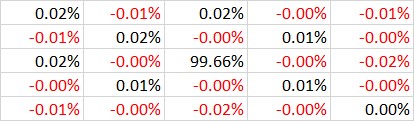}
\includegraphics[scale=0.86]{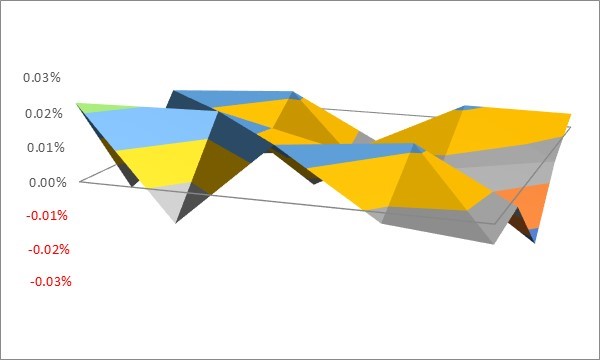}
\caption{Reduced density matrix after a single time-step, where we assume the reduced density matrix starts in a Dirac state, that has first been diagonalized by a random unitary matrix. By diagonalizing the reduced density matrix, the simple form for the time-evolution is no longer evident, and the regular structure of the non-central matrix elements is lost.}\label{not}
\end{figure}
\section{Quantifying Market Precision:}\label{sec_num}
\subsection{Metrics:}
To summarise the arguments made so far we have the following different modelling scenarios:
\begin{itemize}
\item Diagonal reduced density matrix: $\rho_{mkt}(t)$, representing the case where the uncertainty in the price is purely down to a lack of information regarding the future.
\item Diagonal matrices representing traded price operators for standard contracts. Ones where, assuming there is enough information (market reduced density matrix is in a pure state) the price is then determined.
\item Imperfect market trading mechanism represented by a non-diagonal reduced matrix.
\item Non standard contracts, or illiquid assets, mis-aligned with the main market reduced density matrix, represented by non-diagonal matrix operators.
\end{itemize}
In this section we illustrate two measures: $f:\mathcal{M}\rightarrow [0,1]$ (where $\mathcal{M}$ represents the set of density matrices in $\mathbb{C}^N$), which can be used to quantify the degree of market precision. In both cases we require the following:
\begin{enumerate}
\item[{\bf A}] $f(\rho_{mkt}(t))=0$ for diagonal $\rho_{mkt}(t)$. This represents the case where there is no lack of precision. Where, in addition to being diagonal, there is perfect information (ie $\rho_{mkt}=|f_k\rangle\langle f_k|$) the price is fully determined.
\item[{\bf B}] $f(\rho_{mkt}(t))=1$ for a pure state with non-zero variance. In this case, the uncertainty in the traded price outcome is due only to the market mechanism rather than a lack of information.
\end{enumerate}
We present two simple measurements that achieve these requirements. The first is based on the variance in the standard price operator (equation \ref{std_price}). Specifically in definition \ref{varmetric} we look at the ratio in the variance before and after a projection onto one of the eigenvectors for the market reduced density matrix.

A classical density matrix has the same eigenvectors as the traded price operator, meaning that for a classical density matrix: $P_{var}(\rho)=0$. Similarly, for a market in pure state, we have $\rho^i=\rho$. Thus: $P_{var}(\rho)=1$.
\begin{definition}\label{varmetric}
\begin{align}\label{varmetric_eq}
P_{var}(\rho)&=\frac{\max_i\big(Tr[X^2\rho^i]\big)}{Tr[X^2\rho]}\\
\text{Where: }\rho&=\sum_{i=1}^{n\leq N}\rho^i\text{, }\rho^i=\lambda_i|v_i\rangle\langle v_i|\nonumber
\end{align}
\end{definition}
The second is based on the ratio between the Shannon entropy calculated using the probability of finding each potential price $x_i$, and the Von-Neumann entropy of the market reduced density matrix. For a classical density matrix, these 2 measurements are the same, and we find: $P_{ent}(\rho)=0$. For a market in a pure state the Von-Neumann entropy is zero, and we find: $P_{ent}(\rho)=1$.
\begin{definition}\label{entmetric}
\begin{align}
P_{ent}(\rho)&=1-\frac{\rho\log(\rho)}{Tr[\sum_{i=1}^Np_i\log(p_i)]}\\
\text{Where: }p_i&=Tr[|f_i\rangle\langle f_i|\rho]\nonumber
\end{align}
\end{definition}
\subsection{Numerical Simulation Setup:}\label{num_sim}
In this section we use the following basic setup:
\begin{itemize}
\item We use $N=101$, with $x_1=-1$ to $x_{101}=+1$. The choice of $N$ is a balance between increasing the granularity of the data, versus the computational time together with the need for smaller time-steps to ensure accuracy \& stability (see for example \cite{Press} section 19.2).
\item We use a time-step of $\delta t=0.004$, and set $\sigma=0.4$.
\end{itemize}
The starting market reduced density matrix is constructed as follows:
\begin{align}\label{rho_0}
\rho_{mkt}(0)&=\sum_{i=1}^Nq_i^2|f_i\rangle\langle f_i|\text{, }q_i=\bigg(\frac{1}{\sqrt{2\pi\sigma_0^2}}\exp\Big(\frac{-x_i^2}{2\sigma_0^2}\Big)\bigg)^{0.5}\text{, }\sigma_0=0.05
\end{align}
\subsection{Simulation Detail:}\label{ent_lim}
To carry out the simulation, we focus on the metric: \ref{entmetric}. This metric is easier to apply in practice, than the metric based on variance, given by definition \ref{varmetric}, owing to the computational difficulty of extracting eigenvectors from large matrices. The simulation proceeds as follows:
\begin{enumerate}
\item[Step 1)] Start from the classical state given by \ref{rho_0}.
\item[Step 2)] Simulate 5000 steps of the non-classical diffusion, with $\nu_u/\nu_d=0.36$.
\item[Step 3)] At this point we split into 2 separate processes. {\bf Simulation 1} proceeds for 95000 further steps of non-classical diffusion.
\item[Step 4)] {\bf Simulation 2} proceeds for 95000 further steps of classical diffusion ($\nu_u/\nu_d=0$).
\end{enumerate}
\begin{figure}
\includegraphics[scale=0.8]{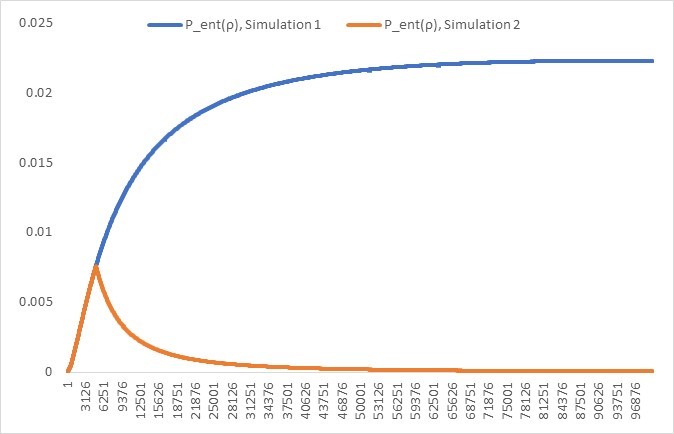}
\caption{The entropy based metric, described by definition \ref{entmetric}, for simulation 1 and simulation 2 as described above}\label{entmetric_chart}
\end{figure}
\begin{figure}
\includegraphics[scale=0.8]{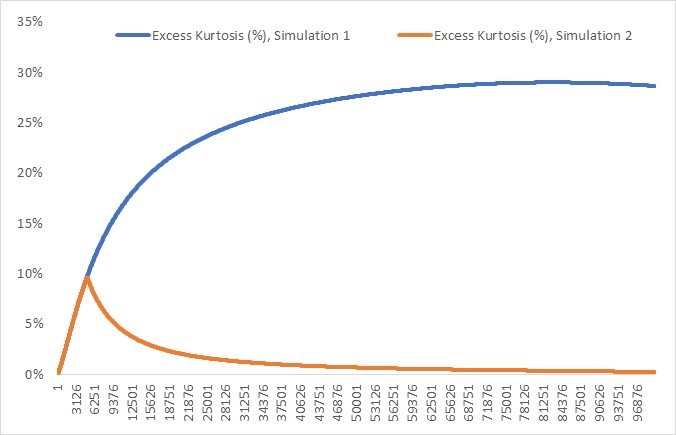}
\caption{The excess kurtosis for simulation 1 and simulation 2 as described above}\label{kurt_chart}
\end{figure}
\subsection{Simulation Results:}
The simulation results for the imprecision metric $P_{ent}(\rho)$ (given by definition \ref{entmetric}) are shown in figure \ref{entmetric_chart}. The resulting excess kurtosis for each simulation is shown in figure \ref{kurt_chart}. The results can be understood as follows:

\noindent\underline{First 5000 Steps:}
\begin{itemize}
\item For the first 5000 steps both simulations follow non-classical diffusion (proposition \ref{nonclass_prop}).
\item The market starts in a classical state. That is, all off-diagonal matrix elements are zero.
\item It therefore follows from proposition \ref{diff_orb} that the off-diagonal sums remain at zero in both simulations.
\item However, due to the non-classical terms described in proposition \ref{nonclass_prop} (for $\nu_u,\nu_d\neq 0$), the individual off-diagonal terms: $\rho_{i(i+j)}\neq 0$ for $j=2,4,\dots$.
\item Note that a matrix element: $\rho_{ii}$, under the nonclassical diffusion spreads to the off-diagonal containing: $\rho_{(i+1)(i-1)}$ and $\rho_{(i-1)(i+1)}$. Thus, whilst the diagonal sums are zero, for even values of $j$ individual matrix elements are non-zero.
\item This is illustrated in fig \ref{OD}, where we show the first 4 nonzero off-diagonals after 5000 time-steps.
\end{itemize}
\underline{Entropy Growth and Misalignment:}
\begin{itemize}
\item It is known (for example see \cite{HicksPHD} proposition 5.3.6) that given fixed classical probabilities: $p_i=Tr[|f_i\rangle\langle f_i|\rho]$, it is the classical density matrix that maximizes the Von-Neumann entropy.
\item The simulation starts at full precision. That is the market is perfectly aligned with/commutes with the main traded price operator. In fact, we have full precision and $P_{ent}(\rho)=0$.
\item However, as the non-zero matrix elements spread away from the diagonal, the Von-Neumann entropy grows at a slower rate than the equivalent classical entropy. Under the non-classical diffusion, precision is lost, and the market becomes {\em misaligned} with the main traded price operator.
\item To put it another way, for a classical process, the uncertainty is solely related to the lack of information. Under the non-classical process, some of the uncertainty is related to the imprecision/market misalignment.
\item Figure \ref{D2_chart} shows the growth in the off-diagonal sum of squared matrix elements:
\begin{align}\label{D2_metric}
||\mathcal{D}_2||&=\sum_{i=1}^{N-2}|\rho_{i(i+2)}|^2
\end{align}
As illustrated in figure \ref{OD}, the magnitude of the individual matrix elements in $\mathcal{D}_2$ grows under the non-classical diffusion, even though the sum of the elements remains at zero. Figure \ref{D2_chart} shows that in simulation 2, the classical diffusion quickly smooths these out, and the $||\mathcal{D}_2||$ metric quickly reduces. In simulation 1, the non-diagonal elements will eventually diffuse to zero (per proposition \ref{stat_pts}), but this takes much longer, with the maximum value for $||\mathcal{D}_2||$ being reached after 32K time-steps.
\end{itemize}
\begin{figure}
\centering
\includegraphics[scale=0.45]{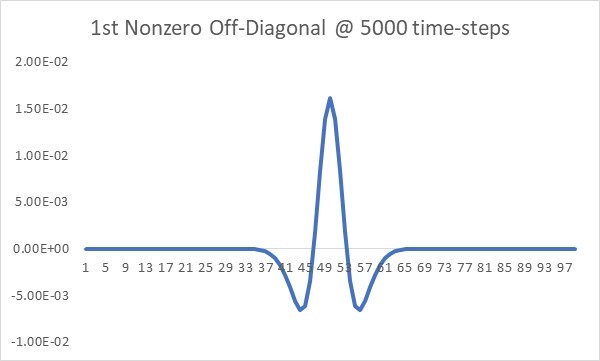}
\includegraphics[scale=0.45]{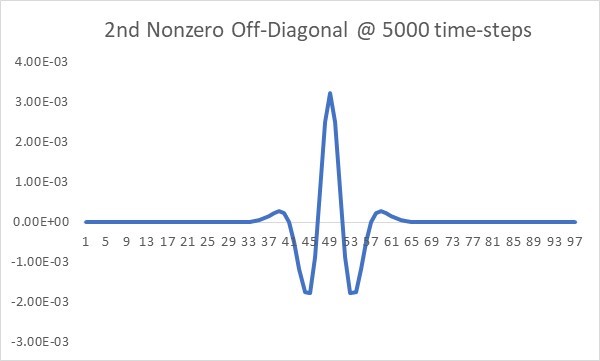}
\includegraphics[scale=0.45]{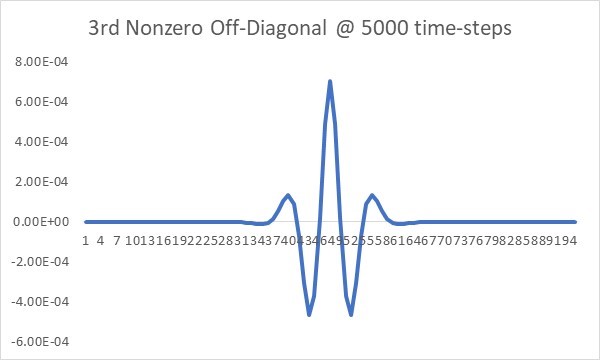}
\includegraphics[scale=0.45]{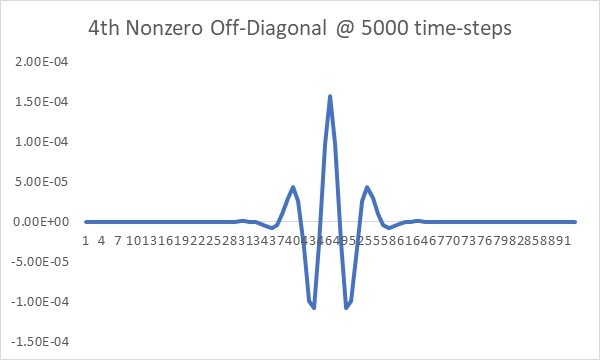}
\caption{The first four non-zero off-diagonals for simulation 1/2 after 5000 steps: $\rho_{i(i+j)},i=1\dots N$ for $j=2,4,6,8$}\label{OD}
\end{figure}
\begin{figure}
\includegraphics[scale=0.8]{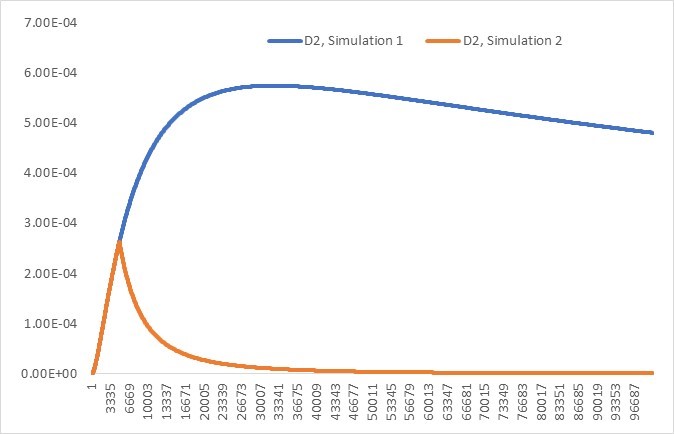}
\caption{The evolution of the $||\mathcal{D}_2||$ metric defined by equation \ref{D2_metric}}\label{D2_chart}
\end{figure}
\underline{Theoretical Entropy Limit:}
\begin{itemize}
\item After 5000 time-steps, we separate out simulation 1, which carries on with the non-classical diffusion, from simulation 2, which now diffuses classically.
\item As per proposition \ref{diff_orb}, the off-diagonal sums are zero.
\item Therefore, per proposition \ref{stat_pts}, under the classical diffusion, the market density matrix will eventually diffuse to the maximum entropy state.
\item That is $S(\rho)\rightarrow\log(N)$, and $P_{ent}(\rho)\rightarrow 0$, as $T\rightarrow\infty$.
\item Figure \ref{frob_dist} shows the distance in the Frobenius norm, from the maximum entropy limit: $\rho=\frac{1}{N}\mathbb{I}$.
\item In both cases the reduced density matrix approaches the maximum entropy limit, but does so at a slower rate in the non-classical case.
\end{itemize}
\begin{figure}
\includegraphics[scale=0.8]{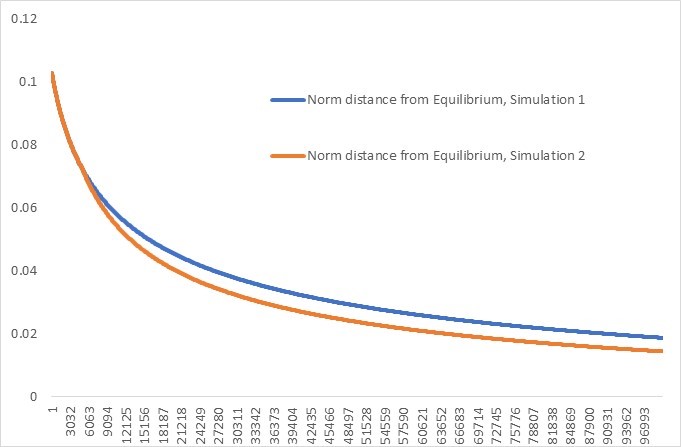}
\caption{Distance in the Frobenius norm from the maximum entropy state.}\label{frob_dist}
\end{figure}
\section{Conclusion:}
Generally, when using classical probability for the study of financial markets, one applies models whereby the uncertainty regarding future prices arises from a lack of information. For example a simple characterization might be that, whilst we know the current price, we don't know what outside events will lead to price increases or decreases in the future. One crucial factor regarding quantum probability is that it allows for uncertainty in what price one can trade at, even where there is perfect information.

Alternatively, as is argued in \cite{HicksPHD} section 5.3.4, one can describe simplified scenarios where markets have varying levels of information, but have the same underlying probabilities for the main traded price. In reality, financial markets do not consist of a single price (like a high street shop) but consist of a complex array of data, a variety of order types and contract types, and even market participants who are withholding their data. Ie, they may be waiting to trade but are yet to submit an order to the exchange. We do not necessarily argue that financial markets are fundamentally `quantum' in nature. Simply that quantum (ie non-commutative) probability provides a useful mathematical framework to represent their dynamics. We seek to build useful models that exhibit some crucial aspect of market behaviour, but do not necessarily make any assumptions regarding the `true' nature of financial market behaviour.

In this article, building on \cite{HicksOQS}, we have started by providing some simple results regarding the long-term behaviour of the Von-Neumann entropy. We also show how one can characterize the different modes of time-evolution by re-writing the models discussed in \cite{HicksPHD} and \cite{HicksOQS} in the standard form.

In section \ref{sec_precis} we go on to discuss the idea that uncertainty in prices can arise from an imperfect market mechanism, or traded prices that are slightly `mis-aligned'. That is, the market may have a fixed price for a small trade in a listed future, but what about a large trade, or a trade with non-standard contract terms. We show how one can design non-standard operators to represent such trades. Indeed, one key aspect of the quantum approach is the way in which one can incorporate the notion of the market misalignment, or alternatively a fundamental limit to the precision with which one can predict prices in advance, into the dynamics as an alternative source of uncertainty additional to growth in the information entropy. 
\section*{Declaration of Conflicting Interests}
The author declares no potential conflicts of interest with respect to the research, authorship, and/or publication of this article.

\end{document}